\documentclass[11pt]{amsart}
\usepackage{geometry}   
\usepackage{natbib}             
\usepackage{graphicx}
\usepackage{amssymb}
\usepackage{graphicx}
\usepackage{epstopdf}
\newcommand{\var}{\mbox{ Var }}
\newtheorem{theorem}{Theorem}
\newtheorem{prop}{Proposition}

\title[Sums of Possibly Associated Bernoulli Variables]{Sums of Possibly Associated Bernoulli Variables: The Conway-Maxwell-Binomial Distribution}
\author{Joseph B. Kadane}
\date{}                                          

\begin{document}
\maketitle

\begin{abstract}
The study of sums of possibly associated Bernoulli random variables has been hampered by an asymmetry between positive correlation and negative correlation.  The Conway-Maxwell Binomial (COMB) distribution and its multivariate extension, the Conway-Maxwell Multinomial (COMM) distribution, gracefully model both positive and negative association. Sufficient statistics and a family of proper conjugate distributions are found. The relationship of this distribution to the exchangeable special case is explored, and two applications are discussed.
 \end{abstract}

\section{Sums of Possibly Associated Bernoulli Variables}\label{sec:sums}

There often are reasons to suggest that Bernoulli random variables, while identically distributed, may not be independent. For example, suppose pots are planted with six seeds each, where each pot has seeds from a unique plant, but different pot's seeds came from different plants. Suppose that success of a seedling is well-defined. If genetic similarity is the dominant source of non-independence, it is reasonable to suppose positive association. However, if competition for nutrients and sunlight predominates, association could be negative.  Hence, it makes sense to find a functional form that gracefully allows for either positive or negative association.

``Association" here means something more general than correlation. Correlation is a particular measure of association, familiar because of its connection with the normal distribution, and its simple relationship to certain expectations. However, there is no particular reason why correlation should be used in non-normal situations if it has undesirable properties.

The desire for a functional form that allows for both positive and negative association runs into the following familiar fact, which is well-known, but for completeness is proved in Appendix A:

\begin{prop}
Suppose $X_1,\ldots, X_m$ have (possibly different) means and variances and common pairwise correlations $\rho$.  Then $\rho \geq -1/(m-1)$. 
\end{prop}
There are (at least) three different possible strategies for dealing with the asymmetry between positive and negative correlation revealed by the proposition:
\begin{itemize}
\item[a)] abandon correlation as a measure of association
\item[b)] abandon exchangeability of the Bernoulli random variables
\item[c)] model the sum directly, without fully specifying the distribution of the underlying Bernoulli random variables.
\end{itemize}

Some light on strategies b) and c) is shed by the following proposition, also proved in Appendix~A.

\begin{prop}\label{prop2}
Let $P\{S=k\} = p_k \geq 0$, where $\sum^m_{k=0} p_k=1$. Then there exists a unique distribution on $X_1,\ldots, X_m$ such that $X_1, \ldots, X_m$ are exchangeable, and $\sum^n_{i=1}X_i$ has the same distribution as does $S$.
\end{prop}

 Proposition~\ref{prop2} is reassuring with respect to strategy c), since the set of distributions on the $X$'s corresponding to an arbitrary distribution on their sum is non-empty. However, it also shows that one can assume exchangeability among the $X$'s without restricting the distribution of their sum, so strategy b) is superfluous. (This fact is also a consequence of \nocite{galambos1978} Galambo's (1978) Theorem 3.2.1.)
 
The distribution studied in this paper pursues strategies a) and c) simultaneously.

There is a voluminous literature on sums of non-independent Bernoulli random variables. An early paper of \citet{Skellam1948} proposed the beta-binomial distribution, a beta mixture of binomials. Thus, the underlying Bernoulli random variables are exchangeable, so this proposal can model positive association, but not negative association (see also \citet{williams1975}). \citet{altham1978} introduces an arithmetic and a multiplicative extension of the binomial distribution, with the intent of modeling non-independence. Both models are exchangeable, and hence are limited in modeling negative association.

\citet{kupper-haseman1978} propose an exchangeable model extending the binomial distribution, based on ideas of \citet{bahadur1961}. Once again the model is exchangeable, so the bounds on the common correlation allow a narrow band of negative correlations that rapidly diminish to zero as the sample size increases. Awkwardly, the parameter constraints depend on the data. S.R. Paul (1985, 1987) \nocite{paul1985,paul1987} proposes two models that aim to unify the beta-binomial and the Kupper/Haseman models. Both suffer from the inevitable narrow range of negative correlations possible, and from data-dependent parameter constraints.

Other papers discussing variations on exchangeable extensions of the binomial include \citet{prentice1986,madsen1993,luceno-deceballos1995,bowman-george1995,george-bowman1995,george-kodell1996,witt2004} and \citet{hisakado-etal2006}.

Additionally, there is a paper discussing the sum of not-identically distributed Bernouilli random variables with a common correlation \citep{gupta-tao2010}. They apply their results in the context of multiple testing. Two papers present two-state Markov models, which of course in equilibrium have identically distributed margins, but are not necessarily exchangeable \citep{viveros-etal1984,rudolfer1990}

By contrast, \citet{ng1989} starts with a completely general class of discrete distributions, defined in a complicated sequential scheme. He then specializes it to the exchangeable case, but in general allows for arbitrary dependent structures.

The remainder of this paper is organized as follows: Section~\ref{sec:cmb} introduces the Conway-Maxwell Binomial distribution and displays some of its mathematical properties. Section~\ref{sec:sufficient} gives sufficient statistics and discusses a conjugate prior family. Section~\ref{sec:understanding}  displays some examples, and gives expressions for its generating functions. The exchangeable case is examined is Section~\ref{sec:exchangeability}, and some applications are shown in Section~\ref{sec:applications}. Appendix C shows that the results given for the COMB distribute extend to its multivariate generalization, the COMM (Conway-Maxwell-Multivariate) distribution.

\section{The Conway-Maxwell Binomial Distribution}\label{sec:cmb}
The binomial distribution, the sum of independent Bernoulli random variables, is extraordinarily useful. Yet there are situations in which the assumption of independence is questionable or unwise. The Conway-Maxwell Binomial distribution (COMB) is a convenient two-parameter family that generalizes the binomial distribution and models both positive and negative association among the Bernoulli summands. 

The probability mass function of the COMB distribution is given by

\begin{equation}\label{eq:one}
P\{W=k\} = \frac{p^k(1-p)^{m-k} \binom{m}{k}^\nu}{S(p, \nu)} \  k=0,1, \ldots, m
\end{equation}
where
$$
S(p,\nu) = \sum^m_{k=0} p^k(1-p)^{m-k}\binom{m}{k}^\nu.
$$
Here $0 \leq p\leq 1$ and $-\infty \leq \nu \leq \infty$ (see \citet[eqn. (13)]{shmueli-etal2005}). Of course, when $\nu =1$, the binomial distribution results.

When $\nu > 1$, the center of the distribution is upweighted relative to the binomial distribution and the tails downweighted. In the limit as $\nu \rightarrow \infty$, $W$ piles up at $m/2$ if $m$ is even, and at $\lfloor m/2\rfloor$ and $\lceil m/2\rceil$ if $m$ is odd.  Thus the component Bernoulli random variables are negatively related. Conversely, when $\nu < 1$, the tails are upweighted relative to the binomial distribution, and the center downweighted. In the limit as $\nu \rightarrow -\infty$, (1) puts all its probability on $W=0$ and $W=m$, which is the extreme case of positive dependence (all $X$'s have the same value). As a consequence, the component Bernoullis are positively related. Thus, $\nu$ measures the extent of positive or negative association in the component Bernoullis. Figure~\ref{fig:new1} (from \citet{kadane-naeshagen2013}) illustrates these points.

\begin{figure}[htbp]
\includegraphics[width=4in]{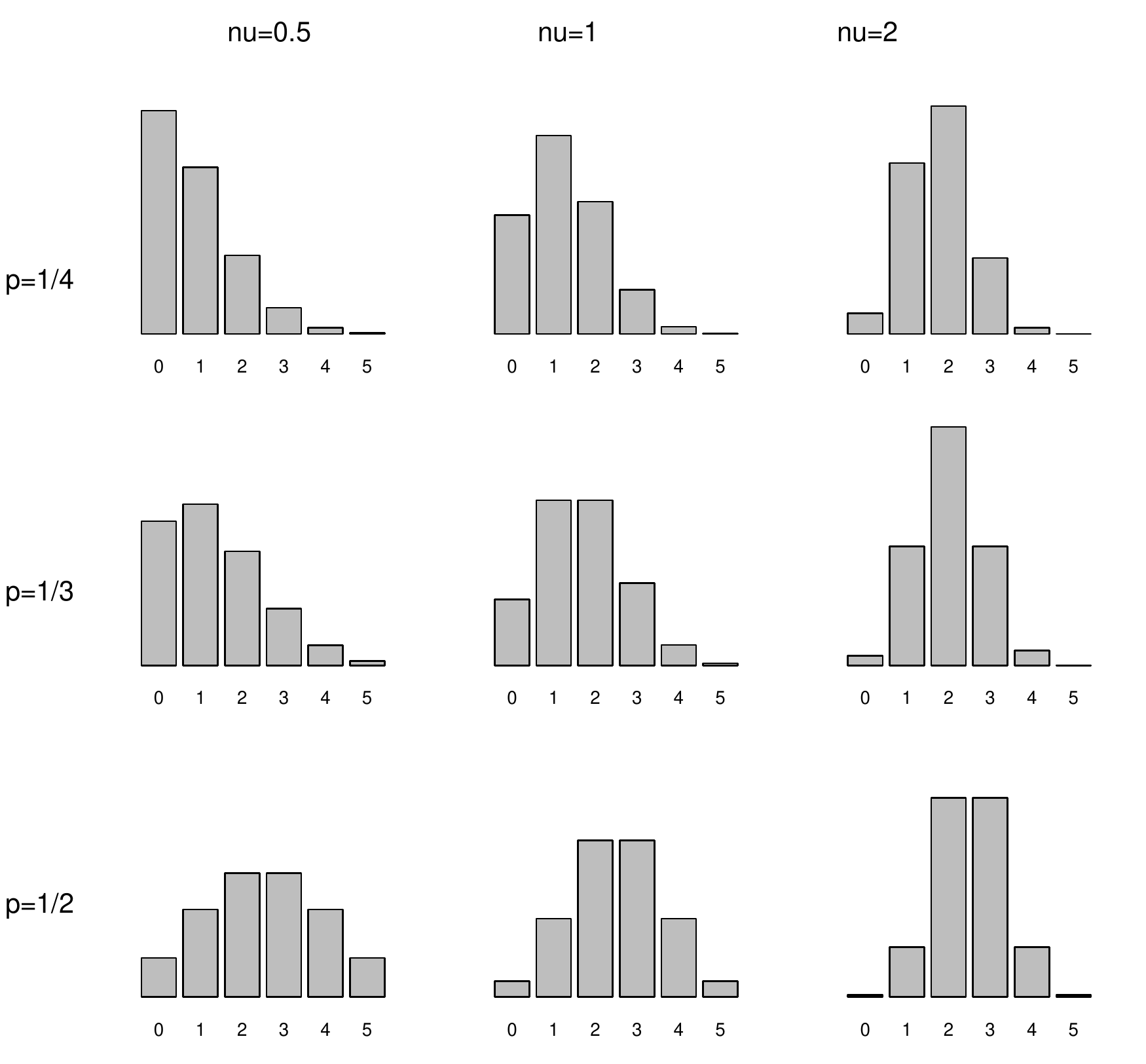}
\caption{Examples of COMB when $m=5$.}\label{fig:new1}
\end{figure}

The name ``Conway-Maxwell" comes from its relationship to the \nocite{conway-maxwell1962} Conway and Maxwell (1962) generalization of the Poisson distribution, COM-Poisson$(\lambda, \nu)$:
\begin{equation}\label{eq:two}
P\{W=x\} = \frac{\lambda^x}{(x!)^\nu \ M(\lambda, \nu)} \ \ x=0,1,\ldots
\end{equation}
where $M(\lambda, \nu) = \sum^\infty_{j=0} \lambda^j/(j!)^\nu$.

\citet{shmueli-etal2005} show that if $X \sim CMP(\lambda_1, \nu)$ and $Y \sim CMP (\lambda_2, \nu)$, $X$ and $Y$ independent, then
\begin{equation}\label{eq:three}
X \mid X + Y \sim COMB\left(\frac{\lambda_1}{\lambda_1 + \lambda_2},\nu\right),
\end{equation}
generalizing the familiar relationship between the Poisson and binomial distributions, when $\nu = 1$.

\section{Sufficient Statistics and a Conjugate Prior Family}\label{sec:sufficient}

Imagine $n$ samples from a COMB distribution, each with respect to a common $m$. Then the likelihood for $p$ and $\nu$ is governed by the data $k_1,\ldots, k_n$, and is given by
\begin{equation}\label{eq:four}
p(k_1,\ldots,k_n \mid p,\nu) = \frac{\prod^n_{i=1} p^{k_i}(1-p)^{m-k_i}\binom{m}{k_i}^\nu}{\left[S(p,\nu)\right]^n}
\end{equation}

The denominator is constant in the data, so it can be ignored. Then
\begin{eqnarray}\label{eq:five}
p(k_1,\ldots, k_n\mid p,\nu) & \propto & (1-p)^{mn}\prod^n_{i=0}\left(\frac{p}{1-p}\right)^{k_i} \frac{m!^{\nu n}}{(k_i!(m-k_i)!)^\nu} \nonumber \\
&\propto & e^{(\sum^m_{i=1} k_i)(\log(p/1-p))-\nu \sum^n_{i=1} \log [k_i!(m-k_i)!]}\nonumber \\
& = & e^{S_1 \log(p/(1-p))-\nu S_2}
\end{eqnarray}
where $S_1 = \sum^n_{i=1}k_i$ and $S_2 = \sum^n_{i=1}\log[k_i!(m-k_i)!]$. Thus the COMB distribution is a member of the exponential family. Consequently, it has a conjugate prior family. To find a convenient form for this family, start over with the likelihood
\begin{equation}\label{eq:six}
p^k(1-p)^{m-k}\binom{m}{k}^\nu \bigg/ \sum^m_{k=0} p^k(1-p)^{m-k} \binom{m}{k}^\nu.
\end{equation}
We may take out the inessential factors of $(1-p)^m(m!)^\nu$, yielding
$$
\left(\frac{p}{1-p}\right)^k \frac{1}{[k!(m-k)!]^\nu} \bigg/ \sum^m_{k=0} \left(\frac{p}{1-p}\right)^k \frac{1}{[k!(m-k)!]^\nu}.
$$

Let $\Psi=\log(p/(1-p))$. Then the likelihood is
\begin{equation}\label{eq:seven}
\frac{e^{\Psi k}}{[k!(m-k)!]^\nu}\bigg/ \sum^m_{k=0} e^{\Psi k}\big/[k!(m-k)!]^\nu,  \ \ k=0,1,\ldots, m.
\end{equation}

Consider a conjugate prior of the form
\begin{eqnarray}\label{eq:eight}
h(\Psi,\nu) & = &  g(\Psi,\nu)e^{\Psi a -b\nu}Z^{-c}(\Psi,\nu)K(a,b,c), \nonumber \\
& &   - \infty < \Psi < \infty, -\infty < \nu < \infty,
\end{eqnarray}
where $Z(\Psi,\nu)  =  \sum^m_{k=0} e^{\Psi k}\big/ [k!(m-k)!]^\nu$\\
and $K^{-1}(a,b,c)  =  \int^\infty_{-\infty} \int^\infty_{-\infty} g(\Psi,\nu)e^{\Psi a-b\nu}Z^{-c}(\Psi,\nu)d \Psi d\nu$.

To maintain the property that the conjugate prior family is closed under sampling, the factor $g(\Psi,\nu)$ can be chosen arbitrarily. However, the propriety of the conjugate family is important for two reasons:
\begin{enumerate}
\item the usual argument for updating  prior to posterior in conjugate families depends on the constant of proportionality being finite.
\item the proper behavior of numerical algorithms for computing posterior distributions, such as grid methods and Markov Chain Monte Carlo, also depend on the propriety of the posterior.
\end{enumerate}

For those reasons, it makes sense to choose $g(\Psi,\nu)$ so that $K^{-1}(a,b,c) < \infty$. There are additional reasons to do so. The function $e^{k x}$ goes to infinity as $x \rightarrow \infty$ if $k$ is positive, and to infinity as $x \rightarrow -\infty$ if $k$ is negative. Hence, if one chose $g(\Psi, \nu) \equiv 1$ in \eqref{eq:eight}, one would be declaring that extreme values of $\Psi$ and $\nu$ are much more likely than others. This would, I judge, be an unusual belief. Instead, I suspect that values of $\nu$ most likely to be of interest are those close to $1$, and values of $\Psi$ perhaps those close to zero. A simple choice expressing these ideas is
\begin{equation}\label{eq:nine}
g(\Psi, \nu) = \o (\Psi) \o(\nu-1)
\end{equation}
where $\o$ is the normal probability density function. Because the normal distribution goes to zero quickly for values far from its mean, this choice has the implication of ``taming" the tails of \eqref{eq:eight}. With this choice, the following theorem results:

\begin{theorem}\label{thm:one}
$K^{-1}(a,b,c) < \infty$.

The proof is in Appendix B.
\end{theorem}

When propriety holds, the updating of \eqref{eq:eight} with data $k$ is given by
$$
a^\prime = a+k, b^\prime = b + \log (k!(m-k)!), \mbox{ and } c^\prime = c+1.
$$

\section{Understanding the COMB distribution}\label{sec:understanding}

One way to understand a distribution is to look at some representative examples of it. Figure~\ref{fig:new1} offers a matrix of such examples, for different values of $p$ and $\nu$.

Another way to understand a distribution is by way of its generating functions. These are derived next. Reconsider
\begin{eqnarray}\label{eq:new10}
S(p, \nu) & = & \sum^m_{k=0} \binom{m}{k}^\nu p^k(1-p)^{m-k} \nonumber\\
& = & (1-p)^m \sum^m_{k=0} \binom{m}{k}^\nu (\frac{p}{1-p})^k\\
& = & (1-p)^m T\left(\frac{p}{1-p},\nu\right). \nonumber
\end{eqnarray}
where $T(x,\nu) = \sum^m_{k=0} x^k\binom{m}{k}^\nu$.

Then the probability generating function of the COMB distribution can be expressed as
\begin{eqnarray}\label{eq:new11}
E(t^x) & = & \sum^m_{k=0} t^k p^k(1-p)^{m-k}\binom{m}{k}^\nu\bigg/S(p,\nu) \nonumber\\
& = & (1-p)^m \sum^m_{k=0} \left(\frac{tp}{1-p}\right)^k \binom{m}{k}^\nu \bigg/ S(p, \nu) \\
& = & T(tp/(1-p), \nu)\big/ T (p/(1-p), \nu). \nonumber
\end{eqnarray}
Similarly, the moment generating function and the characteristic function are, respectively,
\begin{equation}\label{eq:new12}
E (e^{tx}) = T(e^t p\big/(1-p), \nu)\big/T(p/(1-p), \nu)
\end{equation}
and
\begin{equation}\label{eq:new13}
E(e^{itx})= T(e^{it} p\big/(1-p), \nu) / T(p\big/(1-p),  \nu).
\end{equation}

\section{Exchangeability}\label{sec:exchangeability}

The COMB distribution is a distribution on the sum of $m$ (possibly dependent) Bernoulli components without specifying anything else about the joint distribution of those components. This section explores the consequences of assuming in addition that those components are exchangeable.

To establish notation, let
\begin{equation}\label{eq:ten}
p_{i, \ldots, i_m} = P\{X_1 = i_1, X_2 = i_2, \ldots, X_m= i_m\},
\end{equation}
where each $i_j \epsilon \{0,1\}$. Let $\pi$ be a permutation of $(i_i, \ldots, i_m)$. Then the random variables $X$ are called exchangeable just in case
\begin{equation}\label{eq:eleven}
p_{i, \ldots, i_m} = p_{\pi(i_1), \pi(i_2)\ldots, \pi(i_m)}
\end{equation}
for all permutations $\pi$.

Let $S(\ell, m)$ be the set of sequences $(i_i,\ldots, i_m)$ with exactly $\ell$ 1's, {\em i.e.}, satisfying $\sum^m_{j=1}i_j = \ell$. There are $\binom{m}{\ell}$ such sequences in $S(\ell, m)$. The following theorem is given in the literature (see \citet[Theorem 1]{diaconis1977} and the references cited there):

\begin{theorem}\label{theorem:2}
The set $\mathcal{E}_m$ of exchangeable sequences is a convex set whose extreme points are $e_o, \ldots, e_n$, where $e_\ell$ is the measure that puts probability $1/\binom{m}{\ell}$ on each element of $S(\ell,m)$ and $0$ otherwise. Each point $x \epsilon \mathcal{E}_m$ has a unique representation as a mixture of the $m+1$ extreme points.
\end{theorem}

Viewed in this light, the exchangeable COMB distribution specifies a particular two parameter family, with parameters $p$ and $\nu$, of weights on the extreme points $e_o, \ldots, e_m$.

Because $m$-exchangeability applies to every permutation of length $m$, it implies $m^\prime$ exchangeability for each $m^\prime < m$. Hence as $m$ increases, $m$-exchangeability becomes increasingly restrictive. In the limit at $m = \infty$, deFinetti's Theorem shows that sums of exchangeable random variables are mixtures of Binomial random variables. Because the marginal distribution of each component is Bernoulli, interest centers on the joint distribution of pairs of such variables. By exchangeability, every pair has the same distribution as every other pair, so concentrating on $(X_1, X_2)$ suffices. Exchangeability implies 
that $P\{X_1 = 0, X_2=1\} = P\{X_1 = 1, X_2=0\}$, so there are really three probabilities to consider jointly, $p_{00} = P\{X_1 = 0, X_2=0\}$, $p_{01} = p_{10} = P \{X_1 = 0, X_2=1\}$, and $p_{11}=P
\{X_1 = 1, X_2 = 1\}$. \citet[p. 274]{diaconis1977} introduces a convenient way of graphing these quantities. The graph is reminiscent of barycentric co-ordinates, only here the constraint is slightly different:
\begin{equation}\label{eq:twelve}
p_{00} + 2 p_{01} + p_{11} = 1 \ ; \ p_{ij} \geq 0.
\end{equation}
Figures~\ref{fig:one} and \ref{fig:two} display the possible values of the exchangeable COMB distribution for specified values of $m$ and $\nu$, as $p$ varies from 0 to 1.

\begin{figure}
\includegraphics[width=5in]{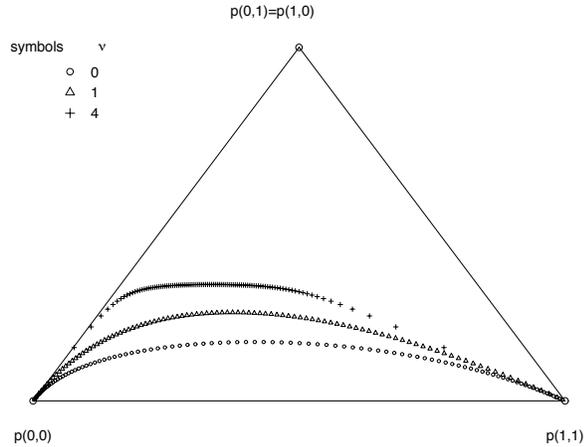}
\caption{Possible values for $P\{X_1 = i, X_2 = j\}$ when $m=3$.}\label{fig:one}
\end{figure}

\begin{figure}
\includegraphics[width=5in]{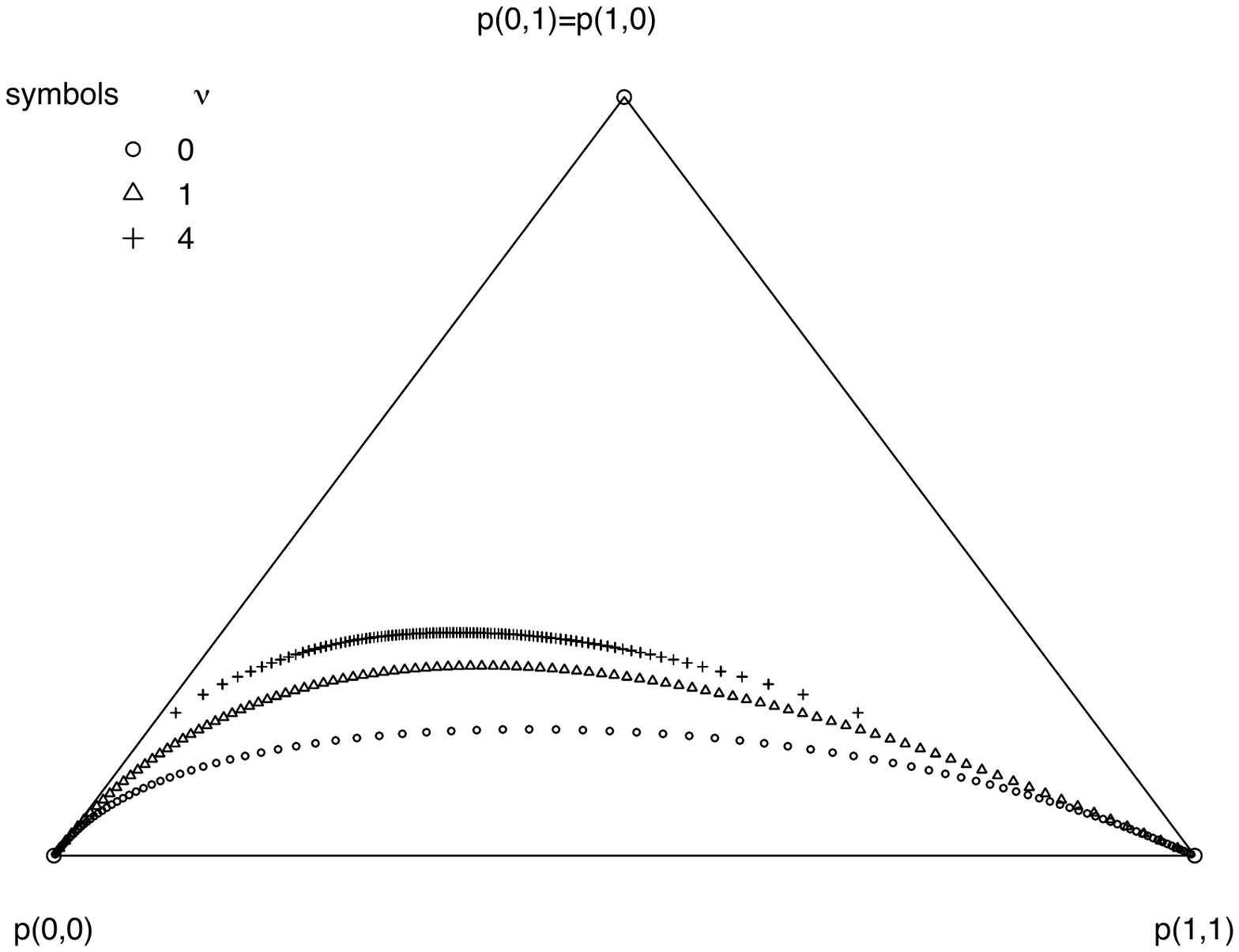}
\caption{Possible values for $P\{X_1 = i, X_2 = j\}$ when $m=5$.}\label{fig:two}
\end{figure}

In Figure~\ref{fig:one}, which is computed at $m=3$, the curve for $\nu = 4$ is the highest, showing, as expected, more weight on $p(0,1) = p(1,0)$. The curve $\nu=1$ is in the middle; this one corresponds to independence, and is known to be $y(1-y)$. The curve for $\nu=0$ is lowest. As $\nu \rightarrow -\infty$, this curve descends to the $p(0,0)$ to $p(1,1)$ lines, indicating that all the probability is at the extremes.

Figure~\ref{fig:two} shows the same curve, when $m=5$. The main difference is that the $\nu=4$ curve is flatter. Indeed, as $m \rightarrow \infty$, this curve will collapse to the $\nu=1$ curve.

\section{Applications}\label{sec:applications} 

\noindent
{a. An agricultural experiment}

\citet{diniz-etal2010} use a correlated binomial model proposed by \citet{luceno-deceballos1995} to analyze data from an experiment on soybean seeds. The model posits summands that are Binomial $(m,p)$, with correlation $\rho$.  They prove that the sum has the same distribution as a mixture of two distributions: with probability $(1-\rho)$ the usual binomial, and with probability $\rho$, a Bernoulli $(p)$ on the points $0$ and $m$. They use an MCMC with data augmentation to fit the model.

The data themselves come from having planted six seedlings in each of 20 pots, and using the judgement of an expert as to which seedlings were successful. The goal was to examine the extent to which competition among the seedings affected the outcomes. The raw data given by \citet{diniz-etal2010} is reported in Table~\ref{tab:one}.

\begin{table}[htbp]
\begin{tabular}{cc}
\# of ``good" plants & observed data\\
0 & 0 \\
1 & 2\\
2 & 2 \\
3 & 5\\
4 & 5\\
5 & 3\\
6 & 3\\
\end{tabular}
\caption{Observed frequency of ``good" plants from \citet{diniz-etal2010}.}\label{tab:one}
\end{table}

To employ the COMB model, I choose to use the prior specified by \eqref{eq:nine}, with $a=b=c=0$. This prior is centered on a Binomial model with $p = 1/2$ (which implies $\Psi = 0)$, which seems reasonable.

The contours of the resulting posterior distribution are shown in Figure~\ref{fig:four}. The maximum posterior point is $\hat{\Psi} = 0.30$ and $\hat{\nu} = 0.54$, with inverse Hessian
$$
\sum = \left(\begin{array}{ll}
0.028 & 0018\\
0.018 & 0063
\end{array} \right).
$$
In view of the elliptical shape of the contours in Figure~\ref{fig:four}, it is reasonable to approximate the posterior with a normal distribution with mean $(\hat{\Psi}, \hat{\nu})$ and covariance $\Sigma$, as would be suggested by the asymptotic distribution of posterior distributions from conditionally independent models.

\begin{figure}
\includegraphics[width=5in]{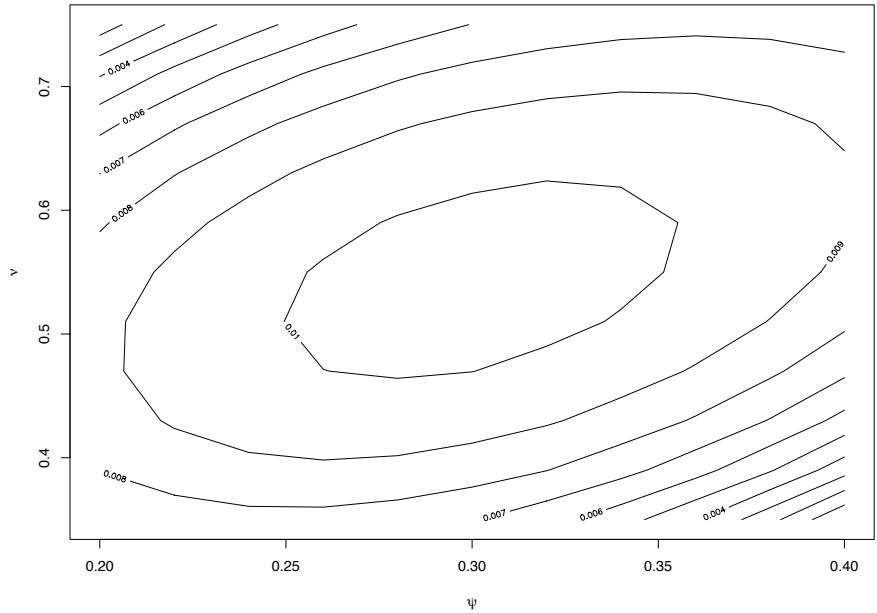}
\caption{Contour plot of the COMB posterior distribution.}\label{fig:four}
\end{figure}

\citet{diniz-etal2010} compare the fit of their model (which they call the ``correlated binomial" (CB)), to that of a binomial distribution.

Extending Table~\ref{tab:one}, Table~\ref{tab:two} below reports the estimated fits of all three models:
\begin{table}[htbp]
\begin{tabular}{ccccc}
\# of ``good" plants & observed data & binomial fit & CB fit & COMB fit\\
0 & 0 & 0.06 & 1.19 & 0.35\\
1 & 2 & 0.61 & 0.79 & 1.24\\
2 & 2 & 2.46 & 2.73 & 2.76\\
3 & 5 & 5.28 & 5.03 & 4.36\\
4 & 5 & 6.37 & 5.21 & 5.04\\
5 & 3 & 4.10 & 2.87 & 4.14\\
6 & 3 & 1.09 & 2.17 & 2.12\\
\end{tabular}
\caption{Fits of various models to the soybean data}\label{tab:two}
\end{table}

The sum of squared errors for the three models are as follows: Binomial 8.96; CB 4.16; COMB 3.77.

It is notable that the COMB estimate of $\nu$ is less than 1, indicating positive association in the soybean seeds. This suggests that competition for nutrients is not the dominant phenomenon in this data set. Further investigation and experimentation might then be warranted to discover the reasons for this positive association. The CB fit did find a positive correlation $\hat\rho = 0.13$. However, the CB model requires $\rho \geq 0$ and hence it could not have found a negative correlation.

In summary, the COMB model offers the following advantages over the CB model:
\begin{enumerate}
\item it allows for both positive and negative association
\item it allows for a conjugate analysis, obviating the need for an MCMC
\item at least for this data set and a squared error metric, it fits better, with the same number of parameters.
\end{enumerate}

\noindent
{\bf b. Killings in Medieval Norway}

In Norway just after the Viking Period, the law distinguished a killing from a murder. In both, there was somebody dead. However, in the former, the killer went to the King's representative within 24 hours and confessed. (Absent such prompt confession, it would be a murder, punishable by execution or banishment). The King's representative would write a letter to the killer stating that the killer was under the protection of the King. An investigation would ensue, resulting in a second letter to the killer, specifying how much was owed to the King, and how much to the family of the deceased. There would then be receipts to the killer  for the payments (two more letters), and a final letter from the King's representative to the killer saying that it was all over. Thus the killer would have received five letters.

Several hundred of these letters have survived in the intervening centuries, and a complete list of those found is available. Additionally, there are mentions of killings in other documents such as private letters, Bishop's records, etc. A simple representation of the data is a $6 \times 2$ matrix, where the first dimension records the number of letters to the killer that survive, and the second is whether or not the killing is mentioned in other sources. Of course, there is the $(0,0)$ cell of killings for which no letters survive and for which there are no other mentions. To estimate this cell, and hence the total number of killings, \citet{kadane-naeshagen2013,kadane-naeshagen2014} resort to a dual-systems estimate.

Since there's no obvious reason why the survival of letters in the killer's archive should be related to whether the killing is mentioned in the other sources, an independence assumption between the two dimensions  seems reasonable. To model the number of letters from a given killing that might survive, a first thought might be a binomial model. However, since all five letters went to the killer, and were likely stored together, at least at first, it is reasonable to suppose that the event of the survival of a given letter to the killer would be positively associated with the event of the survival of the other letters to the same killer. Thus one would expect $\nu \leq 1$ in the COMB, and Kadane and Naeshagen imposed a prior on $\nu$ putting zero probability on the space $\nu \geq 1$. As it happened, the data favors $\nu > 1$, so the posterior piled up at $\nu =1$, the binomial model.

Nonetheless, this was a successful application of the COMB, in that it allowed for (and rejected) what appeared to be the biggest reasonable threat to the model.

\section{Conclusion}\label{sec:conclusion}

The COMB distribution deserves a place in the tool kit of a statistician. Not all Bernoulli random variables are independent, so a one-parameter extension of the binomial distribution, such as the COM-Binomial, may find other useful applications.

\appendix\label{appA}
\section{Proof of Propositions 1 and 2}
Suppose $X_1,X_2,\ldots,X_m$ have the same means and variances, and identical correlations $\rho$. Then $\rho \geq - 1/(m-1)$.
\begin{proof}
Let $Y_i = (X_i - E(X_i))/\sigma(X_i), i=1,\ldots,m$.
Then $Y_1,Y_2,\ldots,Y_m$ satisfy $E(Y_i)=0$ and $\var(Y_i)=1$. Because correlations are unaffected by location and scale changes, they still have  common covariance $\rho$. Now
\begin{eqnarray*}
 0 & \leq & \var(\sum^m_{i=1}Y_i) = E(\sum^m_{i=1}Y_i)^2 - (E(\sum^m_{i=1}Y_i))^2\\
& = & E(\sum Y_i)^2 = E\sum^m_{i=1}Y^2_i +  \raisebox{-4mm}{\mbox{$\displaystyle{\sum^m_{i=1}\sum^m_{j=1}}\atop{i\neq j}$}} E Y_i Y_j\\
& = & m + m(m-1) \rho  
\end{eqnarray*}
from which the desired result follows immediately.
\end{proof}

\begin{proof}
 For each $k$, there are $\binom{m}{k}$ different arrangements of $k$ $i$'s and $m-k$ 0's. Let each of them have probability $p_k/\binom{m}{k}$. Then $P\{\sum^m_{i=1} X_i = k\} = p_k$ and the $X$'s are exchangeable.

To show uniqueness, if the sum of the probabilities of the sequences with exactly $k$ 1's is not $p_k$, the sum condition is violated. If their probabilities are not equal, exchangeability is violated.
\end{proof}

\section{Proof of Theorem}

To obtain an upper bound on $K^{-1}$, a lower bound on $Z$ is needed. I proceed from Jensen's inequality, using the convexity of $\log(x)$ ({\em i.e.}, the second derivative is negative):

Let $q_0,\ldots, q_m$ be arbitrary probabilities that are non-negative and sum to 1. 

Then
$$
\log\left(\sum^m_{k=0} q_k a_k\right) \geq \sum^m_{k=0} q_k \log a_k.
$$

Now
\begin{eqnarray*}
\log Z(\Psi,\nu) & = & \log \sum^m_{k=0} q_ke^{\Psi k}\bigg/q_k(k!(m-k)!)^\nu\\
& \geq & \sum^m_{k=0} q_k \log \left(e^{\Psi k}\bigg/ q_k(k!(m-k)!)^\nu\right)\\
& = & \Psi \sum^m_{k=0}q_k \cdot k - \nu \sum^m_{k=0} q_k \log (k!(m-k)!) - \sum^m_{k=0}q_k \log q_k.
\end{eqnarray*}

Let $Q$ be a random variable on the non-negative integers $\{0,1,\ldots,m\}$ with probability mass $Pr\{Q=k\}=q_k$. Then the bound can be written as
$$
Z(\Psi,\nu) \geq \underline{Z}(\Psi,\nu) = e^{\Psi E(Q) - \nu E(\log (Q!(m-Q)!)}\prod^m_{k=0}q^{q_k}_k
$$
Therefore
\begin{eqnarray*}
K^{-1}(a,b,c) & = & \int^\infty_{-\infty} \int^\infty_{-\infty} g(\Psi,\nu)e^{a \Psi-b \nu} Z^{-c}(\Psi, \nu) d \Psi d \nu\\
& \geq & \int^\infty_{-\infty}\int^\infty_{-\infty} g(\Psi,\nu)e^{a \Psi-b \nu} e^{c\{\Psi E Q - \nu E(\log(Q !(m-Q)!)\}} \prod^m_{k=0} q^{cq_k}_k.
\end{eqnarray*}

Substituting $g(\Psi,\nu) = \o(\Psi)\o(\nu-1)$ and collecting terms
$$
K^{-1}(a,b,c) \geq \int^\infty_{-\infty} \frac{e^{-\Psi^2/2+\Psi(a-cE(Q))-(\nu-1)^2/2 - \nu(b-cE(\log Q!(m-\theta)!))}}{2 \pi}\prod^m_{k=0}q^{cq_k}_k d\Psi d\nu.
$$

Both the integral with respect to $\Psi$ and that with respect to $\nu$ are of the form $e^{-(x^{2}/2)+ xk}$, which are normal integrals, and hence finite. The constant $(\prod^m_{k=0} q^{cq_k}_k)/2\pi$ is also finite. Therefore we have $K^{-1}(a,b,c)  < \infty$, as was to be shown. \hfill $\Box$

Remarks:
\begin{enumerate}
\item This proof uses bounds similar to those in \citet{kadane-etal2006}.
\item This theorem also holds if  instead of $g(\theta, \nu)$ as specified in \eqref{eq:nine}, any other normal distribution for $(\Psi,\nu)$ were used instead.
\end{enumerate}

\section{The Conway-Maxwell-Multivariate Distribution}

The Conway-Maxwell-Multivariate (COMM) Distribution has probability mass  function (for fixed $m$)
$$
P\{\pmb{X} = \pmb{k} \big| (\pmb{P},V)\} = \binom{m}{\pmb{k}}^\nu \prod^r_{i=1} p_i^{k_i}\bigg/
\sum_{{\pmb j}\in D} \binom{m}{\pmb{j}}^\nu \prod^r_{i=1} p_i^{j_i},\pmb{k} \epsilon D
$$
where
$$ 
\begin{array}{rcl}
\pmb{p} &  = & (p_1,\ldots, p_r), p_i \geq 0, \sum^r_{i=1} p_i=1\\
\pmb{k} & = & (k_1,\ldots, k_r), k_i \geq 0, \sum^r_{i=1} k_i = m, k_i's \mbox{ integers}\\
\binom{m}{\pmb{k}} & = & \frac{m!}{k_1(k_2!\ldots k_r!)}
\end{array}
$$
and $D$ is the set of vectors of integers $\pmb{j}$ satisfying $j_i \geq 0$ 
and $\sum^r_{i=1}j_i = m$.

\begin{prop}\label{prop3}
Suppose $X_1,\ldots, X_r$ are independently distributed with probability mass function Conway-Maxwell Poisson $X_i \sim CMP(\lambda_i,\nu)$:
$$
P\{X_i = s_i \mid \lambda_i,\nu\} = \frac{\lambda_i^{s_i}}{(s_i!)^
\nu Z(\lambda_i,\nu)}
$$
where $Z(\lambda_i,\nu) = \sum^\infty_{j=0} \frac{\lambda_i^j}{(j!)^\nu}.$

Then $\pmb{X} \mid \sum^r_{i=1} X_i = m$ has a COMM Distribution with parameters $p_i=\lambda_i/\lambda$ and $\nu$, where $\lambda=\sum^r_{i=1} \lambda_i$.
\end{prop}

\begin{proof}
Let $S = \sum^r_{i=1} X_i,  \pmb{\lambda} = (\lambda, \ldots, \lambda_r)$ and $G(\pmb{p},\nu) = \sum_{\pmb{j} \in D} \binom{m}{\pmb{j}}^\nu \prod^r_{i=1} p_i^{j_i}$.

Then
$$
\begin{array}{rcl}
P\{S=m\} & = & \sum_{\pmb{j}\in D} \prod^r_{i=1} \frac{\lambda_i^{j_i}}{(j_i !)^\nu Z(\lambda_i,\nu)}\\
& = & \frac{\lambda^m}{(m!)^\nu} \prod^r_{i=1} Z(\lambda_i,\nu) \sum_{\pmb{j}\in D} \left(\lambda_i\big/\lambda\right)^{j_i} \binom{m}{\pmb{j}}^\nu\\
& = & \frac{\lambda^m}{(m!)^\nu} \cdot \frac{1}{\prod^r_{i=1}Z(\lambda_i,\nu)} G(\pmb{x}/\lambda,\nu)
\end{array}
$$

Hence
$$
\begin{array}{rcl}
P\{\pmb{X} = \pmb{k}\mid S=m\} & = & \prod^r_{i=1} \frac{\lambda_i^{k_i}}{(k_i!)^\nu} Z(\lambda_i,\nu)
\bigg/\frac{\lambda^m}{(m!)^\nu} \prod^r_{i=1}Z^{G(\pmb{\lambda}/\lambda, \nu)}_{\lambda_i,\nu)}\\
& = & \prod^r_{i=1} (\lambda_i/\lambda)^{k_i}\binom{m}{\pmb{k}}^\nu 
\bigg/ G(\pmb{\lambda}/\lambda,\nu), \mbox{ for } \pmb{k} \in D
\end{array}
$$
which is the probability mass function of the COMM distribution with parameters $p_i=\lambda_i/\lambda (i=1,\ldots, r)$ and $\nu$.
\end{proof}

\begin{prop}\label{prop4}
Let $P\{\pmb{S}=\pmb{k}\} = p_{\pmb{k}}\geq 0$, where $\sum_{\pmb{k}\in D} p_{\pmb{k}} = 1$. Then there exists a unique distribution on $\pmb{X}_1, \pmb{X}_2, \ldots, \pmb{X}_m$ such that $\pmb{X}_1, \pmb{X}_2, \ldots, \pmb{X}_m$ are exchangeable and $\sum^m_{i=1}\pmb{X}_i$ has the same distribution as does $S$.
\end{prop}

\begin{proof}
For each $\pmb{k}\in D$, there are $\binom{m}{\pmb{k}}$ different arrangements of 1's and 0's such that each vector component $i$ has $k_i$ 1's and $(m-k_i)$ 0's. Let each such arrangement have probability $p_{\pmb{k}} \big/\binom{m}{\pmb{k}}$.
Then $P(\sum^r_{i=1}\pmb{X}_i=\pmb{k})=p_{\pmb{k}}$
and the $\pmb{X}_i$'s are exchangeable. To show uniqueness, if the sum of the probabilities of the sequences with $k_i$ 1's in the $i^{th}$ vector component for each $i$ were not $p_{\pmb{k}}$, the sum constraint would not be met. If they did not have equal probability, exchangeability would be violated.
\end{proof}

\begin{prop}\label{prop5}
The COMM distribution has the following sufficient statistics:
$$
\begin{array}{rcl}
S_0 & = & \sum^n_{j=1} \log [k_{ij} ! \ldots k_{rj}!]\\
S_i & = & \sum^n_{j=1} k_{ij}, i = 1, \ldots, r-1
\end{array}
$$
where $k_{ij}$ is the $i^{th}$ component of the $j^{th}$ sample.
\end{prop}

\begin{proof}
$$
\begin{array}{rcl}
p(\pmb{k}_1, \ldots, \pmb{k}_n \mid \pmb{p},\nu) & = & \prod^n_{j=1} \left[\binom{m}{\pmb{k}_j}^\nu \prod^r_{i=1} p_i^{k_{ij}}\bigg/G(\pmb{p})\right]\\
& \propto & p_r^{nm} (m!)^{\nu n} \prod^n_{j=1} \left[\prod^{r-1}_{i=1} (p_i/p_r)^{k_{ij}}\right] \left[\prod^r_{i=1} k_{ij}!\right]^{-\nu}\\
& \propto & e \sum^{r-1}_{i=1} \log(p_i/p_r) \sum^n_{j=1} k_{ij} - \nu \sum^n_{j=1} \log (\prod^r_{i=1} k_{ij}!)\\
& = & e^{\sum^{r-1}_{i=1} \log(p_i/p_r) S_i - \nu S_0}.
\end{array}
$$
\end{proof}

Let $\pmb{\psi}= (\log(p_1/p_r), \log(p_2/p_r), \ldots, \log (p_{r-1}/p_r))$. Consider a conjugate family of the form
$$
h^*(\pmb{\psi},\nu) = g(\pmb{\psi},\nu)e^{-\pmb{\psi}\cdot \pmb{a} - b\nu}G^{-c}(\pmb{\psi},\nu) K(\pmb{a},b,c)
$$
 where $\pmb{a}$ is a vector of length $r-1$, and $b$ and $c$ are positive numbers. Here $G(\pmb{\psi},\nu) = \sum_{{\pmb{j}\in D}} \binom{m}{\pmb{j}}^\nu \prod^{r-1}_{i=1} \psi^{j_i}_i$. If $g(\pmb{\psi},\nu)$ is taken to have a normal distribution (of dimension $r$), then $K^{-1}(\pmb{a},b,c)< \infty$. 

In this case, updating occurs as follows:
$$
\pmb{a}^\prime = \pmb{a} + \pmb{k}^*, b^\prime = b + \log (k_1! k_2! \ldots k_r!), c^\prime = c+1,
$$ 
where $\pmb{k}^* = (k_1,k_2,\ldots,k_{r-1})$.

\begin{proof}
(Generalization of Appendix B).

Let $Q$ be a distribution over $D$ and let $Q^* = (Q_1,\ldots, Q_{r-1})$. Then Jensen's inequality gives
$$
\begin{array}{rcl}
\log \sum_{\pmb{k}\in D} q_k a_j & \geq & \sum_{k \in D} q_k \log a_k\\
\log G(\pmb{\psi},\nu ) & = & \log \sum_{j \in D} q_k \log (e^{\pmb{\psi} \cdot \pmb{k}^*}\bigg/q_k(\prod k_i!)^\nu)\\
& \geq & \sum_{j \in D} q_k \log (e^{\pmb{\psi}\cdot \pmb{k}^*}\bigg/ q_k(\prod k_i !)^\nu)\\
& = & \pmb{\psi} \cdot \sum_{j \in D} q_{\pmb{k}} \pmb{k}^* - \nu \sum_{k \in D} q_k \log(\prod k_i !) - \sum_{j \in D} q_j \log q_j\\
& = & \pmb{\psi} E(\pmb{Q}^*) - \nu E \log(Q_1! Q_2! \ldots Q_r!) - \sum_{j \in D} q_j \log q_j.
\end{array}
$$
Since these are linear in $\pmb{\psi}$ and $\nu$, with any normal prior on $(\pmb{\psi},\nu)$, the integral $K^{-1}$ is finite
\end{proof}

\section*{Acknowledgements} 
I thank Christian Robert, Kim Sellers, Galit Shmueli and Rebecca Steorts for helpful comments.


\end{document}